\documentclass[runningheads]{llncs}
\pdfoutput=1
\usepackage[utf8]{inputenc}
\usepackage{graphicx,xcolor,enumerate}
\usepackage{subfig}
\usepackage{amssymb, amsmath}
\usepackage{xspace}
\usepackage[english]{babel}
\usepackage{cite}

\DeclareMathAlphabet{\mathcal}{OMS}{cmsy}{m}{n}

\newcommand{\degree}{\ensuremath{^{\circ}}\xspace}
\newcommand{\true}{\emph{true}\xspace}
\newcommand{\false}{\emph{false}\xspace}

\newcommand{\maxtotal}{\textsc{MaxTotal}\xspace}
\newcommand{\maxmin}{\textsc{MaxMin}\xspace}

\def\comment#1{}%
\def\withcomments{%
  \newcounter{mycommentcounter}%
   \def\comment##1{\refstepcounter{mycommentcounter}%
    \ifhmode%
     \unskip%
     {\dimen1=\baselineskip \divide\dimen1 by 2 %
       \raise\dimen1\llap{\tiny
	{-\themycommentcounter-}}}\fi%
     \marginpar[{\renewcommand{\baselinestretch}{0.8}%
       \hspace*{-2em}\begin{minipage}{5em}\footnotesize%
[\themycommentcounter]:%
\raggedright ##1\end{minipage}}]{\renewcommand{\baselinestretch}{0.8}%
       \begin{minipage}{5em}\footnotesize%
[\themycommentcounter]: \raggedright%
##1\end{minipage}}}%
  }

\graphicspath{{fig/},{fig/examples/}}

\title{Consistent Labeling of Rotating Maps}
\author{Andreas Gemsa\thanks{Supported by the Concept for the Future of
KIT within the framework of the German Excellence Initiative.}, Martin
N{\"o}llenburg$^\star$, Ignaz Rutter
}
\authorrunning{Andreas Gemsa \and Martin N{\"o}llenburg \and Ignaz Rutter}
\date{}
 \institute{
   Institute of Theoretical Informatics, Karlsruhe Institute of
   Technology (KIT), Germany 
 }

\begin{document}
 \maketitle

\begin{abstract}
  Dynamic maps that allow continuous map rotations, e.g., on mobile
  devices, encounter new issues unseen in static map labeling
  before. We study the following dynamic map labeling problem: The
  input is a static, labeled map, i.e., a set $P$ of points in the
  plane with attached non-overlapping horizontal rectangular labels.
  The goal is to find a \emph{consistent} labeling of $P$ under
  rotation that maximizes the number of visible labels for all
  rotation angles such that the labels remain horizontal while the map
  is rotated. A labeling is consistent if a
  single \emph{active} interval of angles is selected for each label such
  that labels neither intersect each other nor occlude points in $P$ at any rotation
angle.

  We first introduce a general model for labeling rotating maps and
  derive basic geometric properties of consistent solutions. We show
  NP-completeness of the active interval maximization problem even for
  unit-square labels.  We then present a constant-factor approximation
  for this problem based on line stabbing, and refine it further into
  an efficient polynomial-time approximation scheme (EPTAS).  Finally, we extend the EPTAS to the more general setting
  of rectangular labels of bounded size and aspect ratio.
\end{abstract}

\section{Introduction}
Dynamic maps, in which the user can navigate continuously through
space, are becoming increasingly important in scientific and
commercial GIS applications as well as in personal mapping
applications. In particular GPS-equipped mobile devices offer various
new possibilities for interactive, location-aware maps. A common
principle in dynamic maps is that users can pan, rotate, and zoom the
map view. Despite the popularity of
several commercial and free applications, relatively little attention
has been paid to provably good labeling algorithms for dynamic
maps. 

Been et al.~\cite{bdy-dml-06} identified a set of consistency
desiderata for dynamic map labeling. Labels should neither ``jump''
(suddenly change position or size) nor ``pop'' (appear and disappear
more than once) during monotonous map navigation; moreover, the
labeling should be a function of the selected map viewport and not
depend on the user's navigation history. Previous work on the topic
has focused solely on supporting zooming and/or panning of the
map~\cite{bdy-dml-06,nps-dosbl-10,bnpw-oarcd-10}, whereas
consistent labeling under map rotations has not been considered prior
to this paper.

Most maps come with a natural orientation (usually the northern
direction facing upward), but applications such as car or pedestrian
navigation often rotate the map view dynamically to be always forward
facing~\cite{gns-dmcad-05}. Still, the labels must remain horizontally aligned for best
readability regardless of the actual rotation angle of the map. A
basic requirement in static and dynamic label placement is that labels
are pairwise disjoint, i.e., in general not all labels can be placed
simultaneously. For labeling point features, it is further required
that each label, usually modeled as a rectangle, touches the labeled
point on its boundary. It is often not allowed that labels occlude the
input point of another label. Figure~\ref{fig:example} shows an
example of a map that is rotated and labeled. The objective in map
labeling is usually to place as many labels as possible. Translating
this into the context of rotating maps means that, integrated over one
full rotation from $0$ to $2\pi$, we want to maximize the number of
visible labels. The consistency requirements of Been et
al.~\cite{bdy-dml-06} can immediately be applied for rotating maps.

\begin{figure}[tb]
  \centering
  \subfloat{\includegraphics[scale = .8, page=1]{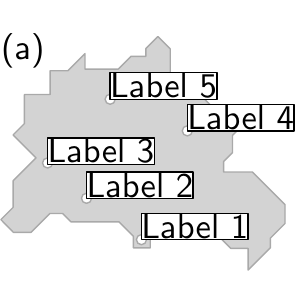}}
  \hfill
  \subfloat{\includegraphics[scale = .8, page=2]{5pts}}
  \hfill
  \subfloat{\includegraphics[scale = .8, page=3]{5pts}}
  \hfill
  \subfloat{\includegraphics[scale = .8, page=4]{5pts}}
  \vspace{-2ex}
  \caption{Input map with five points (a) and three rotated views
    with some partially occluded labels (b)--(d).}
  \label{fig:example}
  \vspace{-2ex}
\end{figure}

\paragraph{Our Results.} Initially, we define a model
for rotating maps and show some basic properties of the different
types of conflicts that may arise during rotation. Next, we prove that
consistently labeling rotating maps is NP-complete, for the
maximization of the total number of visible labels in one full
rotation and NP-hard for the maximization of the visibility range of the least
visible label. Finally, we present a new 1/4-approximation algorithm
and an efficient polynomial-time approximation scheme (EPTAS) for
unit-height rectangles. A PTAS is called \emph{efficient} if its running time is $O(f(\varepsilon) \cdot \text{poly}(n))$. Both algorithms can be extended to the case
of rectangular labels with the property that the ratio of the smallest
and largest width, the ratio of the smallest and largest height, as
well as the aspect ratio of every label is bounded by a constant, even
if we allow the anchor point of each label to be an arbitrary point of
the label. This applies to most practical scenarios where labels
typically consist of few and relatively short lines of text.

\paragraph{Related Work.}
Most previous algorithmic research efforts on automated label
placement cover \emph{static} labeling models for point, line, or area
features. For static point labeling, fixed-position models and slider
models have been introduced~\cite{fw-ppalm-91,ksw-plsl-99}, in which
the label, represented by its bounding box, needs to touch the labeled
point along its boundary. The label number maximization problem is
NP-hard even for the simplest labeling models, whereas there are
efficient algorithms for the decision problem that asks whether all points can
be labeled in some of the simpler models (see the discussion by Klau
and Mutzel~\cite{km-olpfr-03} or the comprehensive map labeling
bibliography~\cite{ws-mlb-96}). Approximation
results~\cite{ksw-plsl-99,aks-lpmis-98},
heuristics~\cite{wwks-3rsgl-01}, and exact
approaches~\cite{km-olpfr-03} are known for many variants of the
static label number maximization problem. 

In recent years, \emph{dynamic} map labeling has emerged as a new
research topic that gives rise to many unsolved algorithmic
problems. Petzold et al.~\cite{pgp-fsmld-03} used a preprocessing step
to generate a reactive conflict graph that represents possible label
overlaps for maps of all scales. For any fixed scale and map region,
their method computes a conflict-free labeling using
heuristics. Mote~\cite{m-fpflp-07} presents another fast heuristic
method for dynamic conflict resolution in label placement that does
not require preprocessing. The consistency desiderata of Been et
al.~\cite{bdy-dml-06} for dynamic labeling (no popping and jumping
effects when panning and zooming), however, are not satisfied by
either of the methods. Been et al.~\cite{bnpw-oarcd-10} showed
NP-hardness of the label number maximization problem in the consistent
labeling model and presented several approximation algorithms for the
problem. Nöllenburg et al.~\cite{nps-dosbl-10} recently studied a
dynamic version of the alternative boundary labeling model, in which
labels are placed at the sides of the map and connected to their
points by leaders. They presented an algorithm to precompute a data
structure that represents an optimal one-sided labeling for all
possible scales and thus allows continuous zooming and panning. None
of the existing dynamic map labeling approaches supports map rotation.

\section{Model}\label{sec:model}
In this section we describe a general model for rotating maps with
axis-aligned rectangular labels. Let $M$ be a labeled input map, i.e.,
a set $P=\{p_1, \dots, p_n\}$ of points  in the plane together with a
set $L = \{\ell_1, \dots, \ell_n\}$ of pairwise disjoint, closed, and axis-aligned
rectangular labels, where each point $p_i$ is a point on the
boundary $\partial \ell_i$ of its label $\ell_i$. We say $\ell_i$ is
\emph{anchored} at~$p_i$. As $M$ is rotated, each label $\ell_i$ in
$L$ remains horizontally aligned and anchored at $p_i$. Thus, label
intersections form and disappear during rotation of $M$. We take the
following alternative perspective on the rotation of $M$. Rather than
rotating the points, say clockwise, and keeping labels horizontal we
may instead rotate each label around its anchor point counterclockwise
and keep the set of points fixed. It is easy to see that both
rotations are equivalent and yield exactly the same results.

A \emph{rotation} of $L$ is defined by a rotation angle $\alpha \in
[0,2\pi)$; a \emph{rotation labeling} of $M$ is a function $\phi\colon L
\times [0,2\pi) \rightarrow \{0,1\}$ such that $\phi(\ell,\alpha) = 1$
if label $\ell$ is visible or \emph{active} in the rotation of $L$ by
$\alpha$, and $\phi(\ell,\alpha) = 0$ otherwise. We call a labeling
$\phi$ \emph{valid} if, for any rotation $\alpha$, the set of labels
$L(\alpha) = \{\ell \in L \mid \phi(\ell,\alpha)=1\}$ consists of
pairwise disjoint labels and no label in $L(\alpha)$ contains any
point in $P$ (other than its anchor point). We note that a valid
labeling is not yet consistent in terms of the definition of Been et
al.~\cite{bdy-dml-06,bnpw-oarcd-10}: given fixed anchor points,
labels clearly do not jump and the labeling is independent of the
rotation history, but labels may still \emph{pop} during a full
rotation from $0$ to $2\pi$, i.e., appear and disappear more than
once. In order to avoid popping effects, each label may be active only
in a single contiguous range of $[0,2\pi)$, where ranges are circular
ranges modulo $2\pi$ so that they may span the input rotation
$\alpha=0$. A valid labeling $\phi$, in which for every label $\ell$
the set $A_\phi(\ell) = \{\alpha \in [0,2\pi) \mid \phi(\ell,\alpha)=1\}$
is a contiguous range modulo $2\pi$, is called a \emph{consistent}
labeling. For a consistent labeling $\phi$ the set $A_\phi(\ell)$ is
called the \emph{active range} of $\ell$. The \emph{length}
$|A_\phi(\ell)|$ of an active range $A_\phi(\ell)$ is defined as the
length of the circular arc $\{(\cos \alpha, \sin \alpha) \mid \alpha
\in A_\phi(\ell)\}$ on the unit circle.

The objective in static map labeling is usually to find a maximum
subset of pairwise disjoint labels, i.e., to label as many points as
possible. Generalizing this objective to rotating maps means that
integrated over all rotations $\alpha \in [0,2\pi)$ we want to display
as many labels as possible. This corresponds to maximizing the sum
$\sum_{\ell \in L} |A_\phi(\ell)|$ over all consistent labelings
$\phi$ of $M$; we call this optimization problem \textsc{MaxTotal}. An
alternative objective is to maximize over all consistent labelings
$\phi$ the minimum length $\min_\ell |A_\phi(\ell)|$ of all active
ranges; this problem is called \textsc{MaxMin}.

\section{Properties of consistent labelings}
\label{sec:properties}
In this section we show basic properties of consistent labelings. If
two labels $\ell$ and $\ell'$ intersect in a rotation of $\alpha$ they
have a (regular) \emph{conflict} at $\alpha$, i.e., in a consistent
labeling at most one of them can be active at $\alpha$. The set
$C(\ell,\ell') = \{\alpha \in [0,2\pi) \mid \ell \text{ and } \ell'
\text{ are in conflict at } \alpha\}$ is called the \emph{conflict
  set} of $\ell$ and $\ell'$.

We show the following lemma in a more general model, in which
the anchor point $p$ of a label $\ell$ can be \emph{any} point within
$\ell$ and not necessarily a point on the boundary $\partial\ell$.

\begin{lemma}\label{lem:4regions}
  For any two labels $\ell$ and $\ell'$ with anchor points $p \in
  \ell$ and $p' \in \ell'$ the set $C(\ell,\ell')$ consists of at most
  four disjoint contiguous conflict ranges.
\end{lemma}

\begin{proof}
  The first observation is that due to the simultaneous rotation of
  all initially axis-parallel labels in $L$, $\ell$ and $\ell'$ remain
  ``parallel'' at any rotation angle $\alpha$. Rotation is a
  continuous movement and hence any maximal contiguous conflict range
  in $C(\ell,\ell')$ must be a closed ``interval'' $[\alpha, \beta]$,
  where $0 \le \alpha, \beta < 2\pi$. Here we explicitly allow $\alpha
  > \beta$ by defining, in that case, $[\alpha,\beta] = [\alpha,2\pi)
  \cup [0,\beta]$. At a rotation of $\alpha$ (resp. $\beta$) the two
  labels $\ell$ and $\ell'$ intersect only on their boundary. Let $l,
  r, t, b$ be the left, right, top, and bottom sides of $\ell$ and let
  $l', r', t', b'$ be the left, right, top, and bottom sides of
  $\ell'$ (defined at a rotation of $0$). Since $\ell$ and $\ell'$ are
  parallel, the only possible cases, in which they intersect on their
  boundary but not in their interior are $t \cap b'$, $b \cap t'$, $l
  \cap r'$, and $r \cap l'$. Each of those four cases may appear
  twice, once for each pair of opposite corners contained in the
  intersection. Figure~\ref{fig:8conflicts} shows all eight boundary
  intersection events. Each of the conflicts defines a unique rotation
  angle and obviously at most four disjoint conflict ranges can be
  defined with these eight rotation angles as their endpoints.
  \begin{figure}[tb]
    \centering
    \subfloat[$r \cap l'$]{\includegraphics[page=2]{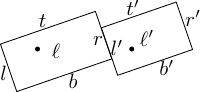}}
    \hfill
    \subfloat[$b \cap t'$]{\includegraphics[page=3]{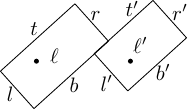}}
    \hfill
    \subfloat[$b \cap t'$]{\includegraphics[page=4]{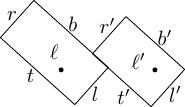}}
    \hfill
    \subfloat[$l \cap r'$]{\includegraphics[page=5]{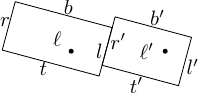}}
    \hfill
    \subfloat[$l \cap r'$]{\includegraphics[page=6]{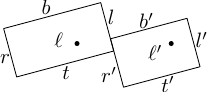}}
    \hfill
    \subfloat[$t \cap b'$]{\includegraphics[page=7]{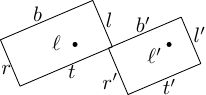}}
    \hfill
    \subfloat[$t \cap b'$]{\includegraphics[page=8]{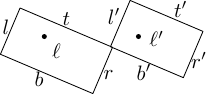}}
    \hfil
    \subfloat[$r \cap l'$]{\includegraphics[page=9]{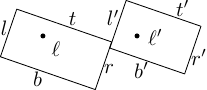}}
    \caption{Two labels $\ell$ and $\ell'$ and their eight possible
      boundary intersection events. Anchor points are marked as black
      dots.}
    \label{fig:8conflicts}
  \end{figure}
\qed
\end{proof}

In the following we look more closely at the conditions under which
the boundary intersection events (also called \emph{conflict events})
occur and at the rotation angles defining them. Let $h_t$ and $h_b$ be
the distances from $p$ to $t$ and $b$, respectively. Similarly, let
$w_l$ and $w_b$ be the distances from $p$ to $l$ and $r$, respectively
(see Figure~\ref{fig:rect-notation}). By $h'_t$, $h'_b$, $w'_l$, and
$w'_r$ we denote the corresponding values for label $\ell'$. Finally,
let $d$ be the distance of the two anchor points $p$ and $p'$. To improve readability of the following lemmas we define two functions $f_d(x) = \arcsin(x/d)$ and $g_d(x) = \arccos(x/d)$.

\begin{figure}[tb]
	\centering
  \includegraphics[page=1]{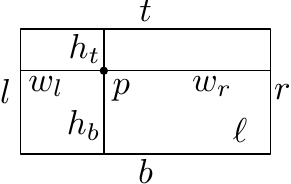}
	\caption{Parameters of label $\ell$ anchored at $p$.}
  \label{fig:rect-notation}
\end{figure}	
	
\begin{lemma}\label{lem:conflicts}
  Let $\ell$ and $\ell'$ be two labels anchored at points $p$ and
  $p'$.  Then the
  conflict events in $C(\ell,\ell')$ are a subset of
  $\mathcal{C} = \{2\pi - f_d(h_t+h'_b), \pi + f_d(h_t+h'_b), f_d(h_b+h'_t), \pi - f_d(h_b+h'_t), 2\pi - g_d(w_r + w'_l), g_d(w_r +
  w'_l), \pi - g_d(w_l + w'_r), \pi + g_d(w_l +
  w'_r)\}$.
\end{lemma}
\begin{proof}
Assume without loss of generality that $p$ and $p'$ lie on a
  horizontal line.
  First we show that the possible conflict events are precisely the
  rotation angles in~$\mathcal{C}$. We start considering the
  intersection of the two sides $t$ and $b'$. If there is a rotation
  angle under which $t$ and $b'$ intersect then we have the situation
  depicted in Figure~\ref{fig:t-bprime} and by simple trigonometric
  reasoning the two rotation angles at which the conflict events occur
  are $2\pi - \arcsin ((h_t+h'_b)/d)$ and $\pi + \arcsin
  ((h_t+h'_b)/d)$. Obviously, we need $d \ge h_t + h'_b$. Furthermore,
  for the intersection in Figure~\ref{fig:t-bprime-1} to be non-empty,
  we need $d^2 \le (w_r + w'_l)^2 + (h_t + h'_b)^2$; similarly, for
  the intersection in Figure~\ref{fig:t-bprime-2}, we need $d^2 \le
  (w_l + w'_r)^2 + (h_t + h'_b)^2$.

\begin{figure}[tb]
	\centering
	\subfloat[rotation of $2\pi -
\alpha$]{\includegraphics[page=2]{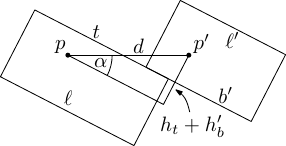}\label{fig:t-bprime-1}}
  \hfil
  \subfloat[rotation of $\pi +
\alpha$]{\includegraphics[page=3]{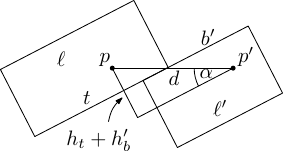}\label{fig:t-bprime-2}}
\caption{Boundary intersection events for $t \cap b'$.}
\label{fig:t-bprime}
\end{figure}

  From an analogous argument we obtain that the rotation angles under
  which $b$ and $t'$ intersect are $\arcsin ((h_b+h'_t)/d)$ and $\pi -
  \arcsin ((h_b+h'_t)/d)$. Clearly, we need $d \ge h_b +
  h'_t$. Furthermore, we need $d^2 \le (w_r + w'_l)^2 + (h_b +
  h'_t)^2$ for the first intersection and $d^2 \le (w_l + w'_r)^2 +
  (h_b + h'_t)^2$ for the second intersection to be non-empty under
  the above rotations.

  The next case is the intersection of the two sides $r$ and $l'$,
  depicted in Figure~\ref{fig:r-lprime}. 
  Here the two rotation angles
  at which the conflict events occur are $2\pi - \arccos
  ((w_r + w'_l)/d)$ and $\arccos ((w_r + w'_l)/d)$. 
  For the first conflict event we
  need $d^2 \le (w_r + w'_l)^2 + (h_t + h'_b)^2$, and for the second
  we need $d^2 \le (w_r + w'_l)^2 + (h_b +   h'_t)^2$.
  For each of the intersections to be non-empty we additionally require 
  that $d \ge  w_r + w'_l$. 


   \begin{figure}[tb]
       \centering
       \subfloat[rotation of $2\pi -
       \beta$\label{fig:r-lprime-1}]{\includegraphics[page=4]{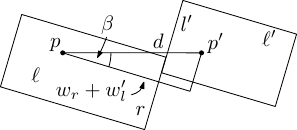}}
       \hfil
       \subfloat[rotation of
       $\beta$\label{fig:r-lprime-2}]{\includegraphics[page=5]{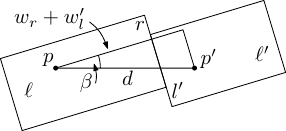}}
  
       \caption{Boundary intersection events for $r \cap l'$.}
       \label{fig:r-lprime}      
   \end{figure}

  Similar reasoning for the final conflict events of $l \cap
  r'$ yields the rotation angles $\pi - \arccos ((w_l + w'_r)/d)$ and
  $\pi + \arccos ((w_l + w'_r)/d)$. The additional constraints are $d
  \ge w_l + w'_r$ for both events and $d^2 \le (w_l + w'_r)^2 + (h_b +
  h'_t)^2$ for the first intersection and $d^2 \le (w_l + w'_r)^2 +
  (h_t + h'_b)^2$. Thus, $\mathcal{C}$ contains all
    possible conflict events.
\qed
\end{proof}

One of the requirements for a valid labeling is that no label may
contain a point in $P$ other than its anchor point. For each label
$\ell$ this gives rise to a special class of conflict ranges, called
\emph{hard} conflict ranges, in which $\ell$ may never be active. The
rotation angles at which hard conflicts start or end are called
\emph{hard} conflict events. Every angle that is a (hard) conflict event 
is called a \emph{label event}.  Obviously, every
hard conflict is also a regular conflict. Regular conflicts that are
not hard conflicts are also called \emph{soft} conflicts. We note that
by definition regular conflicts are symmetric, i.e., $C(\ell, \ell') =
C(\ell',\ell)$, whereas hard conflicts are not symmetric.  The next
lemma characterizes the hard conflict ranges.

\begin{lemma}
  For a label $\ell$ anchored at point $p$ and a point $q\ne p$ in
  $P$, the hard conflict events of $\ell$ and $q$ are a
  subset of $\mathcal{H} = \{2\pi - f_d(h_t),
  \pi + f_d(h_t), f_d(h_b), \pi - f_d(h_b), 2\pi
  - g_d(w_r), g_d(w_r), \pi - g_d(w_l), \pi +
  g_d(w_l)\}$.
\end{lemma}
\begin{proof}
  We define a label of width and height $0$ for $q$, i.e., we set
  $h'_t = h'_b = w'_l = w'_r = 0$. Then the result follows immediately
  from Lemma~\ref{lem:conflicts}.
\qed
\end{proof}

A simple way to visualize conflict ranges and hard conflict ranges is
to mark, for each label $\ell$ anchored at $p$ and each of its (hard) conflict
ranges, the circular 
arcs on the circle centered at $p$ and
enclosing $\ell$. Figure~\ref{fig:markarcs} shows an example. 

\begin{figure}[tb]
\centering
    \includegraphics[page=6]{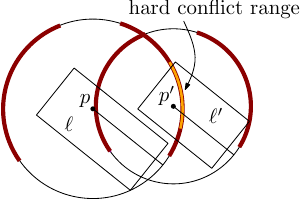}
\caption{Conflict ranges of two labels $\ell$ and $\ell'$
         marked in bold on the enclosing circles.}
\label{fig:markarcs}
\end{figure}

In the following we show that the \maxtotal problem can be discretized
in the sense that there exists an optimal solution whose active ranges
are defined as intervals whose borders are label events. An active
range \emph{border} of a label $\ell$ is an angle $\alpha$ that is
characterized by the property that the labeling $\phi$ is not constant
in any $\varepsilon$-neighborhood of $\alpha$. We call an active range where
both borders are label events a \emph{regular} active range.

\begin{lemma}\label{lem:discretize}
  Given a labeled map $M$ there is an optimal rotation labeling of $M$
consisting of only regular active ranges.
\end{lemma}
\begin{proof}
  Let $\phi$ be an optimal labeling with a minimum number of active
  range borders that are no label events. Assume that there is at
  least one active range border $\beta$ that is no label event. Let
  $\alpha$ and $\gamma$ be the two adjacent active range borders of
  $\beta$, i.e., $\alpha < \beta < \gamma$, where $\alpha$ and $\gamma$
  are active range borders, but not necessarily label events. Then let
  $L_l$ be the set of labels whose active ranges have left border
  $\beta$ and let $L_r$ be the set of labels whose active ranges have
  right border $\beta$. For $\phi$ to be optimal $L_l$ and $L_r$ must
  have the same cardinality since otherwise we could increase the
  active ranges of the larger set and decrease the active ranges of
  the smaller set by an $\varepsilon > 0$ and obtain a better
  labeling.

  So define a new labeling $\phi'$ that is equal to $\phi$ except for
  the labels in $L_l$ and $L_r$: define the left border of the active
  ranges of all labels in $L_l$ and the right border of the active
  ranges of all labels in $L_r$ as $\gamma$ instead of $\beta$. Since
  $|L_l| = |L_r|$ we shrink and grow an equal number of active ranges
  by the same amount. Thus the two labelings $\phi$ and $\phi'$ have
  the same objective value $\sum_{\ell \in L} |A_\phi(\ell)| =
  \sum_{\ell \in L} |A_{\phi'}(\ell)|$. Because $\phi'$ uses as active
  range borders one non-label event less than $\phi$ this number was
  not minimum in $\phi$---a contradiction. As a consequence $\phi$
  has only label events as active range borders.
\qed
\end{proof}

\section{Complexity}

In this section we show that finding an optimal solution for \maxtotal
(and also \maxmin) is NP-hard even if all labels are unit squares and
their anchor points are their lower-left corners. We present a gadget
proof reducing from the NP-complete problem planar
3-SAT~\cite{Lichtenstein1982}. Before constructing the gadgets,
we show a special property of unit-square labels.

\newcommand{\lemunitsquaresrepeattext}{%
  If two unit-square labels $\ell$ and $\ell'$ whose anchor points are
  their lower-left corners have a conflict at a rotation angle
  $\alpha$, then they have conflicts at all angles $\alpha + i\cdot
  \pi/2$ for $i \in \mathbb{Z}$. %
}

\begin{lemma}\label{lem:unitsquaresrepeat}
  \lemunitsquaresrepeattext
\end{lemma}

\begin{proof}
  Similar to the notation used in Section~\ref{sec:properties}, let $f_d =
\arcsin(1/d)$ and $g_d = \arccos(1/d)$.
  From Lemma~\ref{lem:conflicts} we obtain the set $\mathcal{C} =
  \{2\pi -  f_d, \pi +  f_d,  f_d, \pi -
   f_d, 2\pi - g_d, g_d, \pi - g_d, \pi +
g_d\}$ of conflict events for which it is
  necessary that the distance $d$ between the two anchor points is $1
  \le d \le \sqrt{2}$. Since $\arccos x = \pi/2 - \arcsin x$ the set
  $\mathcal{C}$ can be rewritten as  $\mathcal{C} = \{  f_d,
  \pi/2 -  f_d, \pi/2 +  f_d, \pi -  f_d
  , \pi +  f_d, 3\pi/2 -  f_d, 3\pi/2 +  f_d, 2\pi -
 f_d\}$. This shows that conflicts repeat
  after every rotation of~$\pi/2$.
\qed
\end{proof}

For every label $\ell$ we define the \emph{outer circle} of $\ell$ as
the circle of radius $\sqrt{2}$ centered at the anchor point of
$\ell$. Since the top-right corner of $\ell$ traces the outer circle
we will use the locus of that corner to visualize active ranges or
conflict ranges on the outer circle. Note that due to the fact that at
the initial rotation of $0$ the diagonal from the anchor point to the
top-right corner of $\ell$ forms an angle of $\pi/4$ all marked ranges
are actually offset by $\pi/4$.


\subsection{Basic Building Blocks}

\paragraph{Chain.} A \emph{chain} consists of at least four labels
anchored at collinear points that are evenly spaced with
distance~$\sqrt{2}$. Hence, each point is placed on the outer circles
of its neighbors.  We call the first and last two labels of a chain
\emph{terminals} and the remaining part \emph{inner chain}, see
Figure~\ref{fig:buildingblocks:chain}. We denote an assignment of
active ranges to the labels as the \emph{state} of the chain. The
important observation is that in any optimal solution of \maxtotal an
inner chain has only two different states, whereas terminals have
multiple optimal states that are all equivalent for our purposes; see
Figure~\ref{fig:buildingblocks:chain}.  In particular, in an optimal
solution each label of an inner chain has an active range of length
$\pi$ and active ranges alternate between adjacent labels. We will use
the two states of chains as a way to encode truth values in our
reduction.

\newcommand{\leminnerchaintext}{%
In any optimal solution, any label of an inner chain has an active range of
length $\pi$. The active ranges of consecutive labels alternate between
$(0,\pi)$ and $(\pi, 2\pi)$.
}

\begin{lemma}\label{lem:innerchain}
  \leminnerchaintext
\end{lemma}

\begin{proof}
  By construction every label has two hard conflicts at angles $0$ and
  $\pi$, so no active range can have length larger than $\pi$. From
  Lemma~\ref{lem:unitsquaresrepeat} we know that every label has
  conflicts at $\pi/2$ and $3\pi/2$. These conflicts are soft
  conflicts and can be resolved by either assigning all odd labels the
  active range $(0,\pi)$ and all even labels the active range
  $(\pi,2\pi)$ or vice versa. Obviously both assignments are optimal
  and there is no optimal assignment in which two adjacent labels have
  active ranges on the same side of $\pi$.
\qed
\end{proof}

  For inner chains whose distance between two adjacent points is less
  than $\sqrt{2}$ the length of the conflict region changes, but the
  above arguments remain valid for any distance between~$1$ and~$\sqrt{2}$.

\begin{figure}[tb]
  \subfloat[A chain whose two different states are marked as full
  green and dashed blue arcs.]{
\includegraphics[page=1, width=.45\textwidth]{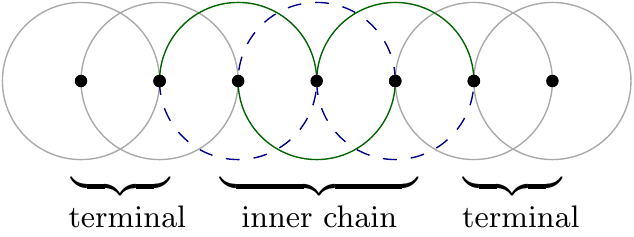}
\label{fig:buildingblocks:chain}
}
\hfill
\subfloat[A turn that splits one inner chain into two inner
chains.]{
\includegraphics[page=1,  width=.45\textwidth]{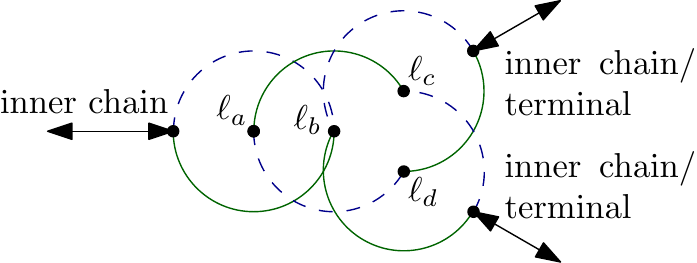}
\label{fig:buildingblocks:turn}
}

\subfloat[Inverter.]{
\includegraphics[page=1,  width=.35\textwidth]{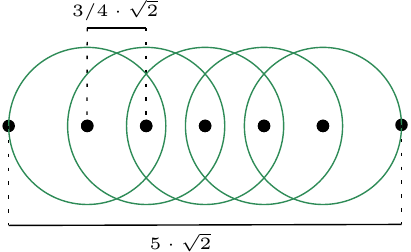}
\label{fig:buildingblocks:inverter}
}
\hfill
\subfloat[Literal Reader.]{
\includegraphics[ width=.65\textwidth]{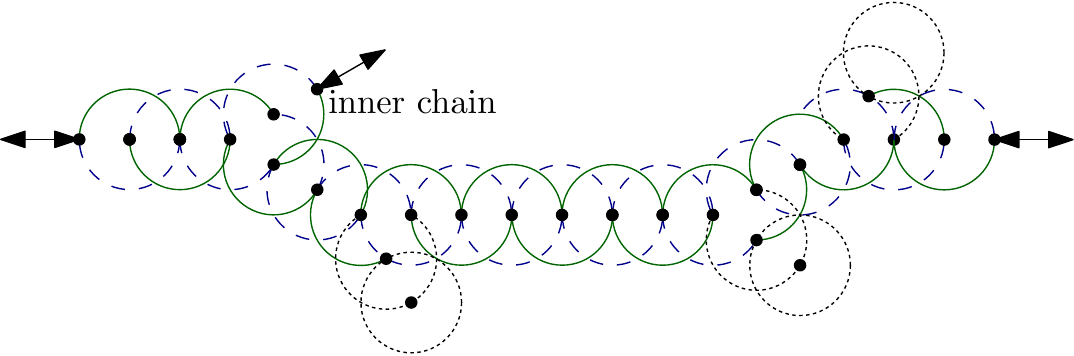}
\label{fig:buildingblocks:var}
}

\label{fig:buildingblocks}
\caption[]{Basic Building Blocks.}
\vspace{-3ex}
\end{figure}

\paragraph{Inverter.} The second basic building block is an
\emph{inverter}.  It consists of five collinear labels that are
evenly spaced with distance~$3/4\cdot\sqrt{2}$ as depicted in
Figure~\ref{fig:buildingblocks:inverter}.  This means that
the five labels together take up the same space as four labels in a
 usual inner chain. Similar to Lemma~\ref{lem:innerchain} the active
 ranges in an optimal solution also alternate. By replacing four labels
 of an inner chain with an inverter we can alter the parity of an inner
 chain.

\paragraph{Turn.} The third building block is a \emph{turn} that
consists of four labels, see Figure~\ref{fig:buildingblocks:turn}.  The anchor
points $p_a$ and
$p_b$ are at distance $\sqrt{2}$ and the pairwise distances between
$p_b$, $p_c$, and $p_d$ are also $\sqrt{2}$ such that the whole
structure is symmetric with respect to the line through $p_a$ and
$p_b$. The central point $p_b$ is called \emph{turn point}, and the
two points $p_c$ and $p_d$ are called \emph{outgoing points}. Due to
the hard conflicts created by the four points we observe that the
outer circle of $p_b$ is divided into two ranges of length $5\pi/6$
and one range of length $\pi/3$. The outer circles of the outgoing
points are divided into ranges of length $\pi$, $2\pi/3$, and
$\pi/3$. The outer circle of $p_a$ is divided into two ranges of
length~$\pi$. The outgoing points serve as connectors to terminals,
inner chains, or further turns. Note, by coupling multiple turns
we can divert an inner chain by any multiple of $30\degree$.

\newcommand{\lemturnsplittext}{%
  A turn has only two optimal states and allows to split an inner
  chain into two equivalent parts in an optimal solution.%
}

\begin{lemma}\label{lem:turnsplit}
  \lemturnsplittext
\end{lemma}

\begin{proof}
  We show that the optimal solution for the turn is $21/6\pi$ and that
  there are only two different active range assignments that yield this
  solution.  Note that for the label $\ell_a$ the length of its active
  range is at most $\pi$.  For $\ell_b$ it is at most $2/3\pi$ and
  for $\ell_c$ and $\ell_d$ it is at most $\pi$.

  We first observe that $\ell_c$ and $\ell_d$ cannot both have an
  active range of length $\pi$ since by
  Lemma~\ref{lem:unitsquaresrepeat} they have a soft conflict in the
  intersection of their length-$\pi$ ranges. Thus at most one of them
  has an active range of length $\pi$ and the other has an active
  range of length at most $5\pi/6$. But in that case the same
  argumentation shows that the active range of $\ell_b$ is at most
  $\pi/2$. Combined with an active range of length $\pi$ for $\ell_a$
  this yields in total a sum of $20\pi/6$. 

  On the other hand, if one of $\ell_c$ and $\ell_d$ is assigned an
  active range of length $2\pi/3$ and the other an active range of
  length $\pi$ as indicated in Figure~\ref{fig:buildingblocks:turn},
  the soft conflict of $\ell_b$ in one of its ranges of length
  $5\pi/6$ is resolved and $\ell_b$ can be assigned an active range of
  maximum length. This also holds for $\ell_a$ resulting in a total
  sum of $21\pi/6$. 



  Since the gadget is symmetric there are only two states that produce
  an optimal solution for the lengths of the active ranges. By
  attaching inner chains to the two outgoing points the truth state of
  the inner chain to the left is transferred into both chains on the
  right.
\qed
\end{proof}

\begin{figure}[tb]
\centering
\subfloat[Clause Gadget with one inner and three outer labels.]{
\includegraphics[page=2, scale=1, trim=5 0 0 0,clip]{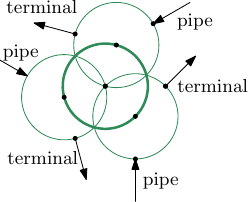}
 \label{appendix:fig:clause_gadget:2}
}
 \hfill
 \subfloat[Clause Gadget with last point of the pipes.]{
 \includegraphics[page=3, scale=1]{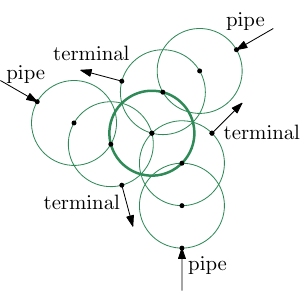}
 \label{appendix:fig:clause_gadget:3}
 }
 \hfill
 \subfloat[Clause Gadget with terminals and last points of the pipes.]{
 \includegraphics[page=4, scale=1]{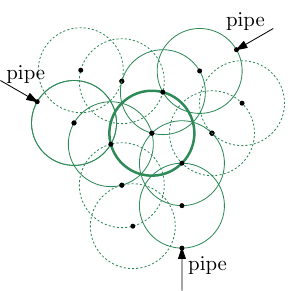}
 \label{appendix:fig:clause_gadget:4}
 }
 \caption{Clause gadget.}\label{fig:clausegadgetappendix}
\end{figure}

\subsection{Gadgets of the Reduction}

\paragraph{Variable Gadget.} 
The variable gadget consists of an alternating sequence of two
building blocks: horizontal chains and \emph{literal
  readers}. A literal reader is a structure that allows us to
split the truth value of a variable into one part running towards a
clause and the part that continues the variable gadget, see
Figure~\ref{fig:buildingblocks:var}. The literal reader consists of
four turns, the first of which connects to a literal pipe and the
other three are dummy turns needed to lead the variable gadget back to
our grid. Note that some of the distances between anchor points in the
literal reader need to be slightly less than $\sqrt{2}$ in order to
reach a grid point at the end of the structure. 

In order to encode truth values we define the state in which the first
label of the first horizontal chain has active range $(0,\pi)$ as
\emph{true} and the state with active range $(\pi,2\pi)$ as
\emph{false}. 

\paragraph{Clause Gadget.}
The clause gadget consists of one inner and three outer labels, where
the anchor points of the outer labels split the outer circle of the
inner label into three equal parts of length $2\pi/3$, see
Figures~\ref{fig:clausegadgetappendix}
and~\ref{fig:clause_gadget:2}. Each outer label further connects
to an incoming literal pipe and a terminal. These two connector labels
are placed so that the outer circle of the outer label is split into
two ranges of length $3\pi/4$ and one range of length $\pi/2$.

The general idea behind the clause gadget is as follows. The inner
label obviously cannot have an active range larger than $2\pi/3$. Each
outer label is placed in such a way that if it carries the value
\false it has a soft conflict with the inner label in one of the three
possible active ranges of length $2\pi/3$.  Hence, if all three labels
transmit the value \false then every possible active range of the
inner label of length $2\pi/3$ is affected by a soft
conflict. Consequently, its active range can be at most $\pi/2$.
 
On the other hand, if at least one of the pipes transmits \true, the
inner label can be assigned an active range of maximum length
$2\pi/3$.

\begin{figure}[tb]
  \begin{minipage}[b]{.35\linewidth}
    \centering
    \includegraphics[page=2, width=.9\textwidth, trim=0 0 0
0,clip]{clause_x.pdf}
  \end{minipage}
  \hfill
  \begin{minipage}[b]{.60\linewidth}
    \centering
    \includegraphics[ width=\textwidth]{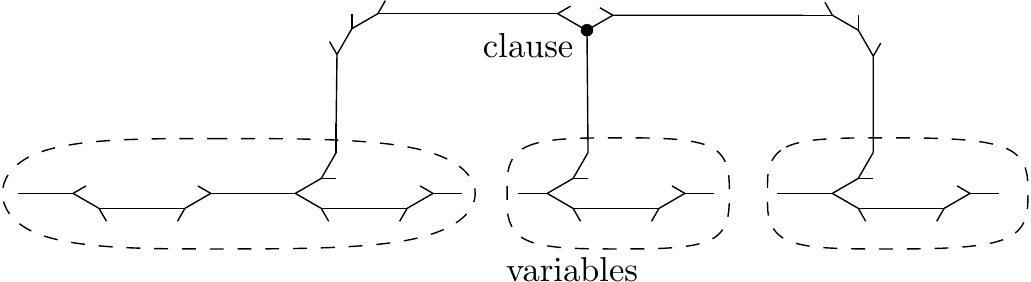}
  \end{minipage}
  \begin{minipage}[t]{.35\linewidth}
    \caption{Clause gadget with one inner and three outer
      labels.}
    \label{fig:clause_gadget:2}
  \end{minipage}
  \hfill
  \begin{minipage}[t]{.60\linewidth}
    \caption{Sketch of the gadget placement for the reduction.}
    \label{fig:maxtotal-proof}
  \end{minipage}
  \vspace{-3ex}
\end{figure}

\newcommand{\lemclausepitext}{%
  There must be a label in a clause or in one of the incoming
  pipes with an active range of length at most $\pi/2$ if and only if
  all three literals of that clause evaluate to \false.
}
\begin{lemma}
\label{lemma:clause:pi}
\lemclausepitext
\end{lemma}

\begin{proof}
  The active range for the lower-right outer label that is equal to
  the state \false is $(3\pi/4, 3\pi/2)$. For the two other outer
  labels the active range corresponding to \false is rotated by $\pm
  2/3\pi$.  Note that the outer clause labels can have an active range
  of at most $3/4\pi$ and the inner clause label can at most have an
  active range of at most $2/3\pi$.  For every literal that is \false
  one of the possible active ranges of the inner clause label is split
  by a conflict into two parts of length $\pi/2$ and $\pi/6$. This
  conflict is either resolved by assigning an active range of length
  $\pi/2$ to the inner clause label or by propagating the conflict
  into the pipe or variable where it is eventually resolved by
  assigning some active range with length at most $\pi/2$.

  Otherwise, if at least one pipe transmits \true, the inner label of
  the clause can be active for $2\pi/3$ while the outer clause labels
  have an active range of length $3\pi/4$ and no chain or turn has a
  label that is visible for less than $2\pi/3$.
\qed
\end{proof}

\paragraph{Pipes.}
Pipes propagate truth values of variable gadgets to clause
gadgets. We use three different types of pipes, which we call
\emph{left arm}, \emph{middle arm}, and \emph{right arm}, depending on
where the pipe attaches to the clause.

One end of each pipe attaches to a variable at the open outgoing label
of a literal reader. Initially, the pipe leaves the variable gadget at
an angle of 30\degree. By using sequences of turns, we can route the
pipes at any angle that is an integer multiple of 30\degree. Thus we
can make sure that for a clause above the variables the left arm
enters the clause gadget at an angle of 150\degree, the middle arm at
an angle of 270\degree, and the right arm at an angle of 30\degree
with respect to the positive $x$-axis. For clauses below the variables
the pipes are mirrored.

In order to transmit the correct truth value into the clause we
first need to place the literal reader such that the turn point of the
first turn corresponds to an even position in the variable
chain. Next, for a positive literal we need a pipe of even length,
whereas for a negative literal the pipe must have odd length. Note
that we can always achieve the correct parity by making use of the
inverter gadgets. 

\paragraph{Gadget Placement.}

We place all variable gadgets on the same $y$-coordinate such that
each anchor point of variable labels (except for literal readers) lies
on integer $x$- and $y$-coordinates with respect to a grid of width
and height $\sqrt{2}$. Clause gadgets and pipes lie below and above
the variables and form three-legged ``combs''. The overall structure
of the gadget arrangement is sketched in
Figure~\ref{fig:maxtotal-proof}.

\begin{theorem}
\maxtotal is NP-complete.
\label{thm:maxtotal:np}
\end{theorem}

\begin{proof}
  For a given planar 3-SAT formula $\varphi$ we construct the \maxtotal
  instance as described above. For this instance we can compute the
  maximum possible sum $K$ of active ranges assuming that each clause
  is satisfiable. By Lemma~\ref{lemma:clause:pi} every unsatisfied
  clause forces one label to have an active range of only
  $\pi/2$. Thus we know that $\varphi$ is satisfiable if and only if
  the \maxtotal instance has a total active range sum of at least
  $K$. Constructing and placing the gadgets can be done in polynomial
  time and space.

  Due to Lemma~\ref{lem:discretize} we can discretize the \maxtotal
  problem. Thus we can construct an oracle that guesses an active
  range assignment, which we can verify in polynomial time. So
  \maxtotal is in $\mathcal{NP}$.
\qed
\end{proof}

We note that the same construction as for the NP-hardness of \maxtotal
can also be applied to prove NP-hardness of \maxmin. The maximally
achievable minimum
length of an active range for a satisfiable formula is $2\pi/3$,
whereas for an unsatisfiable formula the maximally
achievable minimum length is $\pi/2$ due
to Lemma~\ref{lemma:clause:pi}. This observation also yields that
\maxmin cannot be efficiently approximated within a factor of $3/4$.

\begin{corollary}
  \maxmin is NP-hard and it has no efficient approximation
  algorithm with an approximation factor larger than $3/4$ unless
  $\mathcal{P}=\mathcal{NP}$.
\end{corollary}

\section{Approximation Algorithms}
In the previous section we have established that \maxtotal is
NP-complete. Unless $\mathcal{P} = \mathcal{NP}$ we cannot hope for an
efficient exact algorithm to solve the problem. In the following we
devise a $1/4$-approximation algorithm for \maxtotal and refine it to
an EPTAS. 
For both
algorithms we initially assume that labels are congruent unit-height
rectangles with constant width~$w\geq1$ and that the anchor points are the
lower-left corners of the labels.  Let $d$  be the length of the
label's diagonal, i.e., $d=\sqrt{w^2+1}$.

Before we describe the algorithms we state four important properties
that apply even to the more general labeling model, where anchor points
are arbitrary points within the label or on its boundary, and where
the ratio of the smallest and largest width and height, as well as the
aspect ratio are bounded by constants:

\begin{enumerate}[(i)]
\setlength{\itemsep}{1pt}
  \setlength{\parskip}{0pt}
  \setlength{\parsep}{0pt}
\item the number of anchor points contained in a rectangle is
  proportional to its area, 
\item the number of conflicts a label can have with
other labels is bounded by a constant, 
\item any two conflicting labels produce only $O(1)$ conflict
regions, and finally, 
\item there is an optimal \maxtotal solution
where the borders of all active ranges are events.
\end{enumerate}

Properties~(i) and~(ii) are proved in
Lemmas~\ref{lemma:constant_label_in_rectangle}
and~\ref{lemma:constant:conflicts_per_label} using a
simple packing argument. 
Property (iii) follows from property~(ii) and
Lemma~\ref{lem:4regions}. 
Property (iv) follows immediately from
Lemma~\ref{lem:discretize}.

\begin{lemma}\label{lemma:constant_label_in_rectangle}
  For any rectangle $R$ with width~$W$ and height~$H$, the number of
  anchor points in the interior or on the boundary of $R$ is
  proportional to the area of $R$.
\end{lemma}
\begin{proof}
  Recall that by assumption all labels in the initial labeled map $M$
  are visible. Let the smallest label height be $h_{\min}$, the smallest
  label width be $w_{\min}$ and the smallest label area be
  $a_{\min}$. There can be at most $\lceil 2W/w_{\min}\rceil + \lceil
  2H/h_{\min}\rceil$ independent labels intersecting the boundary of $R$
  such that their anchor points are contained in $R$. All remaining
  labels with an anchor point in $R$ must be completely contained in
  $R$, i.e., there can be at most $\lceil W\cdot H / a_{\min}\rceil$
  such labels. Hence, the number of anchor points in $R$ is bounded by
  a constant.
\qed
\end{proof}

\begin{lemma}
  Each label $\ell$ has conflicts with at most a constant number of
  other labels.
\label{lemma:constant:conflicts_per_label}
\end{lemma} 

\begin{proof}
  For two labels $\ell$ and $\ell'$ to have a conflict their outer
  circles need to intersect and thus the maximum possible distance
  between their anchor points is bounded by twice the maximum diameter
  of all labels in $L$. By the assumption that the height ratio, width
  ratio, and aspect ratio of all labels in $L$ is bounded by a
  constant this diameter is constant. Hence we can define for each
  label $\ell$ a constant size area around its anchor point containing
  all relevant anchor points. By
  Lemma~\ref{lemma:constant_label_in_rectangle} this area contains
  only a constant number of anchor points.
\qed
\end{proof}

\subsection{A $1/4$-approximation for \maxtotal}\label{sec:.25-apprx}
The basis for our algorithm is the \emph{line stabbing} or \emph{shifting} technique by Hochbaum
and Maass~\cite{hm-avlsi-85}, which has been applied before to \emph{static} labeling
problems for (non-rotating) unit-height labels~\cite{aks-lpmis-98, ksw-plsl-99}. Consider a grid $G$ where each grid
cell
is a square with side length $2d$. We can address every grid cell by
its row and column index. Now we can partition $G$ into four
subsets by deleting every other row and every other column with either
even or odd parity. Within each of these subsets we have the property
that any two grid cells have a distance of at least $2d$. Thus no two
labels whose anchor points lie in different cells of the same subset
can have a conflict. We say that a grid cell $c$ \emph{covers} a label
$\ell$ if the anchor point of $\ell$ lies inside $c$. By
Lemma~\ref{lemma:constant_label_in_rectangle} only $O(1)$ labels are
covered by a single grid cell. Combining this with
Lemma~\ref{lemma:constant:conflicts_per_label} we see that the number
of conflicts of the labels covered by a single grid cell is
constant. This implies that the number of events in that cell
(cf. Lemma~\ref{lem:discretize}) is also constant.

The four different subsets of grid cells divide a \maxtotal instance
into four subinstances, each of which decomposes into independent grid
cells. If we solve all subsets optimally, at least one of the
solutions is a 1/4-approximation for the initial instance due to the
pigeon-hole principle.

Determining an optimal solution for the labels covered by a grid cell
$c$ works as follows. We compute, for the set of labels $L_c \subseteq
L$ covered by $c$, the set $E_c$ of label events. Due to
Lemma~\ref{lem:discretize} we know that there exists an optimal
solution where all borders of active ranges are label events.  Thus,
to compute an optimal active range assignment for the labels in $L_c$
we need to test all possible combinations of active ranges for all
labels $\ell \in L_c$. For a single cell this requires only constant
time.

%
We can precompute the non-empty grid cells by simple
arithmetic operations on the coordinates of the anchor points and
store those cells in a binary search tree. Since we have $n$ anchor
points there are at most $n$ non-empty grid cells in the tree, and
each of the cells holds a list of the covered anchor points. Building
this data structure takes $O(n \log n)$ time and then optimally
solving the active range assignment problem in the non-empty cells takes $O(n)$ time. 



\begin{theorem}
  There exists an $O(n \log n)$-time algorithm that yields a
  $1/4$-approx\-i\-ma\-tion of \maxtotal for congruent unit-height rectangles with
  their lower-left corners as anchor points. 
\label{thm:const:factor}
\end{theorem}

\subsection{An Efficient Polynomial-Time Approximation Scheme for \maxtotal} 
\label{sec:approx:eptas}
We extend the technique for the
$1/4$-approximation to achieve a $(1 -
\varepsilon)$-approx\-i\-ma\-tion. Let again $G$ be a grid whose grid cells
are squares of side length $2d$. For any integer $k$ we can remove
every $k$-th row and every $k$-th column of the grid cells, starting at
two offsets $i$ and $j$ ($0 \le i,j \le k-1$). This yields collections
of meta cells of side length $(k-1) \cdot 2d$ that are pairwise
separated by a distance of at least $2d$ and thus independent. In
total, we obtain $k^2$ such collections of meta cells. 

For a given $\varepsilon \in (0,1)$ we set $k = \lceil 2/\varepsilon
\rceil$. Let $c$ be a meta cell for the given $k$ and let again $L_c$
be the
set of labels covered by $c$, and $E_c$ the set of label events for
$L_c$. Then, by Lemmas~\ref{lemma:constant_label_in_rectangle}
and~\ref{lemma:constant:conflicts_per_label}, both $|L_c|$ and $|E_c|$
are $O(1/\varepsilon^2)$. 
%
%
Since we need to test all possible active ranges for all
labels in $L_c$, it takes
$O(2^{O(1/\varepsilon^2 \log 1/\varepsilon^2)})$ time
to determine an optimal solution for the meta cell $c$.

For a given collection of disjoint meta cells  we 
determine (as in Section~\ref{sec:.25-apprx}) all $O(n)$ non-empty meta cells and store them in a binary
search tree such that each cell holds a list of its covered anchor
points. This requires again $O(n \log n)$ time. 
So for one collection of meta cells the time complexity for finding an
optimal solution is  $O(n2^{O(1/\varepsilon^2 \log 1/\varepsilon^2)} +
n\log n)$. There are $k^2$ such collections and, by the pigeon hole principle,
the optimal solution for at least one of them is a
$(1-\varepsilon)$-approximation of the original instance. This yields the
following theorem.


\begin{theorem}
  There exists an 
  EPTAS
   that
  computes a $(1-\varepsilon)$-approximation of \maxtotal for congruent
  unit-height rectangles with their lower-left corners as anchor
  points. Its time complexity is $O((n 2^{O(1/\varepsilon^2 \log
    1/\varepsilon^2)} + n\log n)/\varepsilon^2)$.
\end{theorem}

We note that this EPTAS basically relies on properties (i)--(iv) and
that there is nothing special about congruent rectangles anchored at
their lower-left corners. Hence we can generalize the algorithm to the
more general labeling model, in which the ratio of the label heights,
the ratio of the label widths, and the aspect ratios of all labels are
bounded by constants. Furthermore, the anchor points are not required
to be label corners; rather they can be any point on the boundary or
in the interior of the labels. Finally, we can even ignore the
distinction between hard and soft conflicts, i.e., allow that anchor
points of non-active labels are occluded. Properties (i)--(iv)
still hold in this general model. The only change in the EPTAS is to
set the width and height of the grid cells to twice the maximum
diameter of all labels in $L$.

\begin{corollary}
  There exists an 
  EPTAS that
  computes a $(1-\varepsilon)$-approximation of \maxtotal in the
  general labeling model with rectangular labels of bounded height
  ratio, width ratio, and aspect ratio, where the anchor point of each label is an
  arbitrary point in that label. The time
complexity of the EPTAS is $O((n 2^{O(1/\varepsilon^2 \log 1/\varepsilon^2)} +
  n\log n)/\varepsilon^2)$.
\end{corollary}

\section{Conclusion}
\label{sec:conclusion}
We have introduced a new model for consistent labeling of rotating
maps and proved NP-hardness of the active range maximization
problem. We could, however, show that there is an EPTAS for the
\maxtotal problem that works for rectangular labels with arbitrary
anchor points and bounded height ratio, width ratio, and aspect
ratio. An interesting open question and an important challenge in
practice is to combine map rotation with zooming and panning and study
the arising algorithmic labeling problems.

{ \small
\setlength{\itemsep}{0pt}
\bibliographystyle{abuser}
 \bibliography{abbrv,label-rot,maplab}

\begin{thebibliography}{10}

\bibitem{aks-lpmis-98}
P.~K. Agarwal, M.~van Kreveld, and S.~Suri.
\newblock Label placement by maximum independent set in rectangles.
\newblock {\em Comput. Geom. Theory Appl.} 11:209--218, 1998,
  \href{http://dx.doi.org/10.1016/S0925-7721(98)00028-5}%
{doi:10.1016/S0925-7721(98)00028-5}.

\bibitem{bdy-dml-06}
K.~Been, E.~Daiches, and C.~Yap.
\newblock Dynamic map labeling.
\newblock {\em IEEE Transactions on Visualization and Computer Graphics}
  12(5):773--780, 2006, \href{http://dx.doi.org/10.1109/TVCG.2006.136}%
{doi:10.1109/TVCG.2006.136}.

\bibitem{bnpw-oarcd-10}
K.~Been, M.~N{\"o}llenburg, S.-H. Poon, and A.~Wolff.
\newblock Optimizing active ranges for consistent dynamic map labeling.
\newblock {\em Comput. Geom. Theory Appl.} 43(3):312--328, 2010,
  \href{http://dx.doi.org/10.1016/j.comgeo.2009.03.006}%
{doi:10.1016/j.comgeo.2009.03.006}.

\bibitem{fw-ppalm-91}
M.~Formann and F.~Wagner.
\newblock A packing problem with applications to lettering of maps.
\newblock {\em Proc. 7th Annuual ACM Sympos. on Computational Geometry
  (SoCG'91)}, pp.~281--288, 1991,
  \href{http://dx.doi.org/10.1145/109648.109680}%
{doi:10.1145/109648.109680}.

\bibitem{gns-dmcad-05}
E.~Gervais, D.~Nussbaum, and J.-R. Sack.
\newblock Dynamap: a context aware dynamic map application.
\newblock {\em Proc. GISPlanet, Estoril, Lisbon, Portugal}, 2005.

\bibitem{hm-avlsi-85}
D.~S. Hochbaum and W.~Maass.
\newblock {Approximation schemes for covering and packing problems in image
  processing and VLSI}.
\newblock {\em Journal of the ACM} 32(1):130--136, January 1985,
  \href{http://dx.doi.org/10.1145/2455.214106}%
{doi:10.1145/2455.214106}.

\bibitem{km-olpfr-03}
G.~W. Klau and P.~Mutzel.
\newblock Optimal labeling of point features in rectangular labeling models.
\newblock {\em Mathematical Programming (Series B)} pp.~435--458, 2003,
  \href{http://dx.doi.org/10.1007/s10107-002-0327-9}%
{doi:10.1007/s10107-002-0327-9}.

\bibitem{ksw-plsl-99}
M.~van Kreveld, T.~Strijk, and A.~Wolff.
\newblock Point labeling with sliding labels.
\newblock {\em Comput. Geom. Theory Appl.} 13:21--47, 1999,
  \href{http://dx.doi.org/10.1016/S0925-7721(99)00005-X}%
{doi:10.1016/S0925-7721(99)00005-X}.

\bibitem{Lichtenstein1982}
D.~Lichtenstein.
\newblock Planar formulae and their uses.
\newblock {\em SIAM J. Comput.} 11(2):329--343, 1982,
  \href{http://dx.doi.org/10.1137/0211025}%
{doi:10.1137/0211025}.

\bibitem{m-fpflp-07}
K.~D. Mote.
\newblock Fast point-feature label placement for dynamic visualizations.
\newblock {\em Information Visualization} 6(4):249--260, 2007,
  \href{http://dx.doi.org/10.1057/palgrave.ivs.9500163}%
{doi:10.1057/palgrave.ivs.9500163}.

\bibitem{nps-dosbl-10}
M.~N{\"o}llenburg, V.~Polishchuk, and M.~Sysikaski.
\newblock Dynamic one-sided boundary labeling.
\newblock {\em Proc. 18th ACM SIGSPATIAL International Conference on Advances
  in Geographic Information Systems}, pp.~310--319. ACM Press, November 2010,
  \href{http://dx.doi.org/10.1145/1869790.1869834}%
{doi:10.1145/1869790.1869834}.

\bibitem{pgp-fsmld-03}
I.~Petzold, G.~Gr{\"o}ger, and L.~Pl{\"u}mer.
\newblock Fast screen map labeling---data-structures and algorithms.
\newblock {\em Proc. 23rd Internat. Cartographic Conf. (ICC'03)}, pp.~288--298,
  2003.

\bibitem{wwks-3rsgl-01}
F.~Wagner, A.~Wolff, V.~Kapoor, and T.~Strijk.
\newblock Three rules suffice for good label placement.
\newblock {\em Algorithmica} 30(2):334--349, 2001,
  \href{http://dx.doi.org/10.1007/s00453-001-0009-7}%
{doi:10.1007/s00453-001-0009-7}.

\bibitem{ws-mlb-96}
A.~Wolff and T.~Strijk.
\newblock {The Map-Labeling Bibliography}, 1996,
  \url{http://i11www.iti.kit.edu/map-labeling/bibliography}.

\end{thebibliography}

}

\end{document}